\documentclass[11pt,final,onecolumn]{IEEEtran}

\usepackage{amsmath}
\usepackage{amsfonts}
\usepackage{mathrsfs}
\usepackage{enumerate}
\usepackage{mathtools}
\usepackage{amssymb}
\usepackage{graphicx}
\usepackage{color}

\usepackage[bottom=0.8in]{geometry}

\newtheorem{proposition}{Proposition}
\newtheorem{corollary}{Corollary}

\newtheorem{remark}{Remark}
\newcommand {\dfn} {:=}
\newcommand {\reals} {\ensuremath{\mathbb{R}}}
\newcommand {\bx} {\mbox{\boldmath $x$}}
\newcommand {\bE} {\ensuremath{\mathbb{E}}}
\newcommand{\calX}{{\cal X}}
\newcommand {\hP} {\hat{P}}
\newcommand {\hH} {\hat{H}}

\allowdisplaybreaks

\begin{document}
\thispagestyle{empty}
\setcounter{page}{1}
\setlength{\baselineskip}{1.15\baselineskip}

\title{\huge{An Integral Representation of the Logarithmic Function
with Applications in Information Theory\footnote{Published
in the Entropy journal, vol.~22, no.~1, paper~51, pp.~1--23,
January 2020.
\newline Available at: https://www.mdpi.com/1099-4300/22/1/51.}\\[0.8cm]}}

\author{Neri Merhav \qquad Igal Sason\\[0.3cm]
The Andrew and Erna Viterbi Faculty of Electrical Engineering\\
Technion -- Israel Institute of Technology \\
Technion City, Haifa 3200003, Israel\\
E-mail: \{merhav, sason\}@ee.technion.ac.il}

\maketitle
\thispagestyle{empty}

\begin{abstract}
We explore a well--known integral representation of the logarithmic function,
and demonstrate its usefulness in obtaining compact, easily--computable exact
formulas for quantities that involve expectations and higher moments of the
logarithm of a positive random variable (or the logarithm of a sum of i.i.d.
positive random variables).
The integral representation of the logarithm is proved useful in a
variety of information--theoretic applications, including universal lossless
data compression, entropy and differential entropy evaluations, and the calculation
of the ergodic capacity of the single-input, multiple--output (SIMO) Gaussian
channel with random parameters (known to both transmitter and receiver).
This integral representation and its variants are anticipated to serve as
a useful tool in additional applications, as a rigorous alternative
to the popular (but non--rigorous) replica method (at least in some situations).\\

\end{abstract}

\break

\section{Introduction}

In analytic derivations pertaining to many problem areas in information
theory, one frequently encounters the need to calculate expectations and
higher moments of expressions that involve the logarithm of a
positive--valued random variable, or more generally, the logarithm of the sum
of several i.i.d. positive random variables. The common practice, in such situations, is
either to resort to upper and lower bounds on the desired expression (e.g., using Jensen's
inequality or any other well--known inequalities), or to apply the Taylor series
expansion of the logarithmic function. A more modern approach is to use the replica
method (see, e.g., \cite[Chap.\ 8]{MM09}), which is a popular (but non--rigorous) tool
that has been borrowed from the field of statistical physics with considerable success.

The purpose of this work is to point out to an alternative approach and to
demonstrate its usefulness in some frequently--encountered
situations. In particular, we consider the following
integral representation of the logarithmic function (to be proved in the
sequel),
\begin{align}
\label{ir}
\ln x =\int_0^\infty\frac{e^{-u}-e^{-ux}}{u} \; \mathrm{d}u, \quad x > 0.
\end{align}
The immediate use of this representation is in situations where the argument
of the logarithmic function is a positive--valued random variable, $X$, and we
wish to calculate the expectation, $\bE\{\ln X\}$.
By commuting the expectation operator with the integration over $u$ (assuming
that this commutation is valid), the calculation of $\bE\{\ln X\}$ is
replaced by the calculation of the (often easier) moment--generating function
(MGF) of $X$, as
\begin{align}
\label{ElnX}
\bE\{\ln X\}=
\int_0^\infty\left[e^{-u}-\bE\{e^{-uX}\}\right] \; \frac{\mathrm{d}u}{u}.
\end{align}
Moreover, if $X_1,\ldots,X_n$ are positive i.i.d. random variables, then
\begin{align}
\label{Elnsum}
\bE\{\ln(X_1+\ldots+X_n)\}=\int_0^\infty
\left(e^{-u}-\bigl[\bE\{e^{-uX_1}\}\bigr]^n\right) \; \frac{\mathrm{d}u}{u}.
\end{align}
This simple idea is not quite new. It has been used in the physics
literature, see, e.g., \cite[Exercise~7.6, p.~140]{MM09},
\cite[Eq.~(2.4) and onward]{EN93} and \cite[Eq.~(12) and onward]{SSRCM19}.
With the exception of \cite{RajanT15}, we are not aware of any work in the
information theory literature where it has been used. The purpose of this
paper is to demonstrate additional information-theoretic applications, as
the need to evaluate logarithmic expectations is not rare at all in many
problem areas of information theory. Moreover, the integral representation
\eqref{ir} is useful also for evaluating higher moments of $\ln X$, most
notably, the second moment or variance, in order to assess the statistical
fluctuations around the mean.

We demonstrate the usefulness of this approach in several application
areas, including: entropy and differential entropy evaluations, performance
analysis of universal lossless source codes, and calculations of the ergodic
capacity of the Rayleigh single-input multiple-output (SIMO) channel. In some
of these examples, we also demonstrate the calculation of variances associated
with the relevant random variables of interest. As a side remark, in the same
spirit of introducing integral representations and applying them, Simon and
Divsalar \cite{S98, SD98} have brought to the attention of communication
theorists useful, definite--integral forms of the $Q$--function (Craig's
formula \cite{Craig91}) and Marcum $Q$--function, and demonstrated
their utility in applications.

It should be pointed out that most of our results remain in the form of a
single-- or double-- definite integral of certain functions that depend on the
parameters of the problem in question. Strictly speaking, such a definite
integral may not be considered a closed--form expression, but nevertheless,
we can say the following:
\begin{enumerate}[a)]
\item In most of our examples, the expression we obtain is more compact, more
elegant, and often more insightful than the original quantity.
\item The resulting definite integral can actually be considered a closed--form
expression ``for every practical purpose'' since definite integrals in one or
two dimensions can be calculated instantly using built-in numerical integration
operations in MATLAB, Maple, Mathematica, or other mathematical software tools.
This is largely similar to the case of expressions that include
standard functions (e.g., trigonometric, logarithmic, exponential functions, etc.),
which are commonly considered to be closed--form expressions.
\item The integrals can also be evaluated by power series expansions of the integrand,
followed by term--by--term integration.
\item Owing to Item~(c), the asymptotic behavior in the parameters of the model can
be evaluated.
\item At least in two of our examples, we show how to pass from an $n$--dimensional
integral (with an arbitrarily large $n$) to one or two--dimensional integrals.
This passage is in the spirit of the transition from a multi--letter expression to
a single--letter expression.
\end{enumerate}

To give some preliminary flavor of our message in this work, we conclude
this introduction by mentioning a possible use of the integral representation
in the context of calculating the entropy of a Poissonian
random variable. For a Poissonian random variable, $N$, with parameter
$\lambda$, the entropy (in nats) is given by
\begin{align}
H(\lambda)=-\bE\biggl\{\ln\biggl(\frac{e^{-\lambda}\lambda^N}{N!}\biggr)\biggr\}
=\lambda-\lambda\ln\lambda+\bE\{\ln N!\},
\end{align}
where the non--trivial part of the calculation is associated with the last
term, $\bE\{\ln N!\}$. In \cite{EvansB88}, this term was handled by using a
non--trivial formula due to Malmst\'en (see \cite[pp.~20--21]{EMOT81}), which
represents the logarithm of Euler's Gamma function in an integral form
(see also \cite{Martinez07}). In Section~\ref{mathbg}, we derive the
relevant quantity using \eqref{ir}, in a simpler and more transparent form
which is similar to \cite[(2.3)--(2.4)]{Knessl98}.

The outline of the remaining part of this paper is as follows.
In Section~\ref{mathbg}, we provide some basic mathematical background
concerning the integral representation \eqref{ElnX} and some of its
variants. In Section~\ref{appl}, we present the application examples.
Finally, in Section~\ref{outlook}, we summarize and provide some outlook.

\section{Mathematical Background}
\label{mathbg}

In this section, we present the main mathematical background associated with
the integral representation \eqref{ir}, and provide several variants of this
relation, most of which are later used in this paper. For reasons that will
become apparent shortly, we extend the scope to the complex plane.

\begin{proposition}
\label{Proposition 1}
\begin{align} \label{log int. rep.}
\ln z = \int_0^\infty \frac{e^{-u} - e^{-uz}}{u} \; \mathrm{d}u,
\quad \mathrm{Re}(z) \geq 0.
\end{align}
\end{proposition}

\begin{proof}
\begin{align}
\ln z &= (z-1) \int_0^1 \frac{\mathrm{d}v}{1+v(z-1)}\\
\label{int_a2311}
&= (z-1) \, \int_0^1 \int_0^\infty e^{-u[1+v(z-1)]} \, \mathrm{d}u \, \mathrm{d}v\\
\label{int_0212}
&= (z-1) \, \int_0^\infty e^{-u} \, \int_0^1 e^{-uv(z-1)} \, \mathrm{d}v \, \mathrm{d}u \\
&= \int_0^\infty \frac{e^{-u}}{u} \, \Bigl[1 - e^{-u(z-1)} \Bigr] \, \mathrm{d}u \\
&= \int_0^\infty \frac{e^{-u} - e^{-uz}}{u} \; \mathrm{d}u,
\end{align}
where \eqref{int_a2311} holds since $\mathrm{Re}\{1+v(z-1)\} > 0$ for
all $v\in(0,1)$, based on the assumption that $\mathrm{Re}(z) \geq 0$;
\eqref{int_0212} holds by switching the order of integration.
\end{proof}

\begin{remark}
\label{Remark 1}
In \cite[p.~363, Identity~(3.434.2)]{GR14}, it is stated that
\begin{align}
\label{eq:GR}
\int_0^{\infty} \frac{e^{-\mu x} - e^{-\nu x}}{x} \, \mathrm{d}x = \ln \frac{\nu}{\mu},
\qquad \mathrm{Re}(\mu) > 0, \; \mathrm{Re}(\nu) > 0.
\end{align}
Proposition~\ref{Proposition 1} also applies to any purely imaginary number, $z$,
which is of interest too (see Corollary~\ref{Corollary 1} in the sequel, and the
identity with the characteristic function in \eqref{011219}).
\end{remark}

\vspace*{0.1cm}
Proposition~\ref{Proposition 1} paves the way to obtaining some additional related integral
representations of the logarithmic function for the reals.

\begin{corollary}(\cite[p.~451, Identity~3.784.1]{GR14})
\label{Corollary 1}
For every $x>0$,
\begin{align} \label{eq: cor. 1}
\ln x = \int_0^{\infty} \frac{\cos(u) - \cos(ux)}{u} \; \mathrm{d}u.
\end{align}
\end{corollary}
\begin{proof}
By Proposition~\ref{Proposition 1} and the identity
$\ln x \equiv \mathrm{Re}\bigl\{\ln(ix)\bigr\}$ (with $i\dfn\sqrt{-1}$), we get
\begin{align}
\label{0112a}
\ln x &= \int_0^\infty  \frac{e^{-u} - \cos(ux)}{u} \; \mathrm{d}u.
\end{align}
Subtracting both sides by the integral in \eqref{0112a}
for $x=1$ (which is equal to zero) gives~\eqref{eq: cor. 1}.
\end{proof}

\vspace{0.2cm}
Let $X$ be a real-valued random variable, and let
$\Phi_X(\nu)\dfn\bE\bigl\{e^{i\nu X}\bigr\}$ be the characteristic function of $X$.
Then, by Corollary~\ref{Corollary 1},
\begin{align} \label{011219}
\bE\bigl\{\ln X\bigr\} = \int_0^\infty \frac{\cos(u)-
\mathrm{Re}\{\Phi_X(u)\}}{u} \; \mathrm{d}u,
\end{align}
where we are assuming, here and throughout the sequel, that the expectation operation
and the integration over $u$ are commutable, i.e., Fubini's theorem applies.

Similarly, by returning to Proposition~\ref{Proposition 1} (confined to a real--valued
argument of the logarithm), the calculation of $\bE\{\ln X\}$ can be replaced by the
calculation of the MGF of $X$, as
\begin{align}  \label{1}
\bE\{\ln X\}=
\int_0^\infty \Bigl[e^{-u}-\bE\bigl\{e^{-uX}\bigr\}\Bigr] \; \frac{\mathrm{d}u}{u}.
\end{align}
In particular, if $X_1,\ldots,X_n$ are positive i.i.d. random variables, then
\begin{align}
\label{1b}
\bE\{\ln(X_1+\ldots+X_n)\} = \int_0^\infty \Bigl[e^{-u}-\bigl(\bE\bigl\{e^{-uX_1}\bigr\}\bigr)^n \Bigr] \;
\frac{\mathrm{d}u}{u}.
\end{align}

\begin{remark}
\label{Remark 2}
One may further manipulate \eqref{1} and \eqref{1b} as follows.
Since $\ln x \equiv \frac1s \, \ln (x^s)$ for any $s\ne 0$ and $x > 0$, then
the expectation of $\ln X$ can also be represented as
\begin{align}  \label{1-s}
\bE\{\ln X\}
= \frac1s \int_0^\infty \Bigl[e^{-u}-\bE\bigl\{e^{-uX^s}\bigr\} \Bigr] \;
\frac{\mathrm{d}u}{u}, \quad s \neq 0.
\end{align}
The idea is that if, for some $s \notin \{0,1\}$, $\bE\{e^{-uX^s}\}$ can be expressed
in closed form, whereas it cannot for $s=1$ (or even $\bE\{e^{-uX^s}\} < \infty$
for some $s \notin \{0,1\}$, but not for $s=1$), then \eqref{1-s} may prove useful.
Furthermore, if $X_1, \ldots, X_n$ are positive i.i.d. random variables,
$s> 0$, and
$Y = (X_1^s + \ldots + X_n^s)^{1/s}$, then
\begin{align}  \label{1b-s}
\bE\{\ln Y\}= \frac1s \int_0^\infty
\Bigl[e^{-u}- \bigl(\bE\bigl\{e^{-uX_1^s}\bigr\}\bigr)^n\Bigr] \; \frac{\mathrm{d}u}{u}.
\end{align}
For example, if $\{X_i\}$ are i.i.d. standard Gaussian random variables and
$s=2$, then \eqref{1b-s} enables to calculate the expected value of the
logarithm of a chi-squared distributed random variable with $n$ degrees of freedom.
In this case,
\begin{align}
\bE\{e^{-uX_1^2}\} &= \frac1{\sqrt{2\pi}} \int_{-\infty}^{\infty} e^{-ux^2} \,
e^{-x^2/2} \, \mathrm{d}x  \nonumber \\[0.1cm]
&= \frac1{\sqrt{2u+1}},
\end{align}
and, from \eqref{1b-s} with $s=2$,
\begin{align}
\bE\{\ln Y\}= \tfrac12 \int_0^\infty \left[e^{-u}-
(2u+1)^{-n/2}\right] \; \frac{\mathrm{d}u}{u}.
\end{align}
It should be noted that according to the {\em pdf}
of a chi-squared distribution, one can express $\bE\{\ln Y\}$ as a
one-dimensional integral even without using \eqref{1b-s}.
However, for general $s > 0$, the direct calculation of
$\bE\bigl\{\ln \left(\sum_{i=1}^n |X_i|^s \right) \bigr\}$
leads to an $n$-dimensional integral, whereas \eqref{1b-s}
provides a one-dimensional integral whose integrand involves
in turn the calculation of a one-dimensional integral too.
\end{remark}

Identity \eqref{ir} also proves useful when one is interested, not only in
the expected value of $\ln X$, but also
in higher moments, in particular, its second moment or variance.
In this case, the one--dimensional integral becomes a
two--dimensional one. Specifically, for any $s>0$,
\begin{align}
\mathrm{Var}\{\ln X\}
&=\bE\{\ln^2(X)\}-[\bE\{\ln X\}]^2\\[0.1cm]
&=\frac1{s^2} \, \bE\left\{\int_0^\infty \int_0^\infty
\bigl(e^{-u}-e^{-uX^s}\bigr) \, \bigl(e^{-v}-e^{-vX^s}\bigr)
\; \frac{\mathrm{d}u \, \mathrm{d}v}{uv}\right\} \nonumber\\[0.1cm]
& \hspace*{0.4cm} -\frac1{s^2} \, \int_0^\infty \int_0^\infty
\bigl( e^{-u}-\bE\{e^{-uX^s}\} \bigr) \,
\bigl(e^{-v}-\bE\{e^{-vX^s}\}\bigr)
\; \frac{\mathrm{d}u \, \mathrm{d}v}{uv}\\[0.1cm]
&= \frac1{s^2} \int_0^\infty \int_0^\infty
\left[\bE\bigl\{e^{-(u+v)X^s}\bigr\}
-\bE\bigl\{e^{-uX^s}\bigr\} \,
\bE\bigl\{e^{-vX^s}\bigr\} \right]
\; \frac{\mathrm{d}u \, \mathrm{d}v}{uv}\\[0.1cm]
&=\frac1{s^2} \int_0^\infty \int_0^\infty
\mathrm{Cov}\bigl\{e^{-uX^s},e^{-vX^s}\bigr\}
\; \frac{\mathrm{d}u \, \mathrm{d}v}{uv}.
\end{align}
More generally, for a pair of positive random variables,
$(X,Y)$, and for $s>0$,
\begin{align}
\mathrm{Cov}\{\ln X,\ln Y\}
= \frac1{s^2} \int_0^\infty\int_0^\infty\mathrm{Cov}\bigl\{e^{-uX^s},e^{-vY^s}\bigr\}
\; \frac{\mathrm{d}u \, \mathrm{d}v}{uv}.
\end{align}

For later use, we present the following variation of the basic identity.
\begin{proposition}
\label{Proposition 2}
Let $X$ be a random variable, and let
\begin{align}
\label{g}
M_X(s)\dfn \bE\bigl\{e^{sX}\bigr\}, \quad \forall \, s \in \reals,
\end{align}
be the MGF of $X$. If $X$ is non-negative, then
\begin{align} \label{2511a}
& \bE\bigl\{ \ln(1+X)\bigr\} = \int_0^\infty \frac{e^{-u} \, [1 - M_X(-u)]}{u}
\; \mathrm{d}u,
\end{align}
and
\begin{align} \label{2511b}
& \mathrm{Var}\bigl\{ \ln(1+X) \bigr\} = \int_0^\infty \int_0^\infty
\frac{e^{-(u+v)}}{uv} \;
\Bigl[ M_X(-u-v) - M_X(-u) \, M_X(-v) \Bigr] \, \mathrm{d}u \, \mathrm{d}v.
\end{align}
\end{proposition}

\begin{proof}
Equation \eqref{2511a} is a trivial consequence of \eqref{1}.
As for \eqref{2511b}, we have
\begin{align}
&\mathrm{Var}\bigl\{\ln(1+X)\bigr\} \nonumber \\[0.1cm]
&= \bE\bigl\{\ln^2(1+X)\bigr\} - \Bigl( \bE\bigl\{\ln(1+X)\bigr\} \Bigr)^2
\nonumber \\[0.1cm]
&= \bE\biggl\{ \int_0^{\infty} \frac{e^{-u}}{u} \; \bigl(1 - e^{-uX} \bigr) \,
\mathrm{d}u
\int_0^{\infty} \frac{e^{-v}}{v} \; \bigl(1 - e^{-vX} \bigr) \, \mathrm{d}v
\biggr\} \nonumber \\[0.1cm]
& \hspace*{0.5cm} - \int_0^\infty \int_0^\infty \frac{e^{-(u+v)} \, [1 -
M_X(-u)] \, [1-M_X(-v)]}{uv} \; \mathrm{d}u \, \mathrm{d}v \\[0.1cm]
&= \int_0^\infty \int_0^\infty \frac{e^{-(u+v)}}{uv} \; \bE\Bigl\{ \bigl(1 -
e^{-uX} \bigr) \, \bigl(1 - e^{-vX} \bigr) \Bigr\}
\, \mathrm{d}u \, \mathrm{d}v \nonumber \\[0.1cm]
& \hspace*{0.5cm} - \int_0^\infty \int_0^\infty \frac{e^{-(u+v)}
\, [1 - M_X(-u) -M_X(-v) + M_X(-u) \, M_X(-v)]}{uv} \; \mathrm{d}u \,
\mathrm{d}v \\[0.1cm]
& = \int_0^\infty \int_0^\infty \frac{e^{-(u+v)}}{uv} \; \Bigl[ 1 - M_X(-u) -
M_X(-v) + M_X(-u-v) \Bigr]
\, \mathrm{d}u \, \mathrm{d}v \nonumber \\[0.1cm]
& \hspace*{0.5cm} - \int_0^\infty \int_0^\infty \frac{e^{-(u+v)}}{uv} \;
\Bigl[1 - M_X(-u) - M_X(-v) + M_X(-u) \, M_X(-v) \Bigr]
\; \mathrm{d}u \, \mathrm{d}v \\[0.1cm]
& = \int_0^\infty \int_0^\infty \frac{e^{-(u+v)}}{uv} \; \Bigl[ M_X(-u-v) -
M_X(-u) \, M_X(-v) \Bigr] \mathrm{d}u \, \mathrm{d}v.
\end{align}
\end{proof}

The following result relies on the validity of \eqref{log int. rep.} to the
right-half complex plane, and its derivation is based on the identity
$\ln(1+x^2) \equiv \ln(1+ix) + \ln(1-ix)$ for all $x \in \reals$.
In general, it may be used if the characteristic function of a random variable
$X$ has a closed--form expression, whereas the MGF of $X^2$ does not admit a
closed--form expression (see Proposition~\ref{Proposition 2}).
We introduce the result, although it is not directly used in the paper.

\begin{proposition}
\label{Proposition 3}
Let $X$ be a real-valued random variable, and let
\begin{align}
\Phi_X(u) := \bE\bigl\{e^{iuX}\bigr\}, \quad \forall \, u \in \reals,
\end{align}
be the characteristic function of $X$. Then,
\begin{align} \label{2411a}
\bE\bigl\{\ln(1+X^2)\bigr\} = 2 \int_0^\infty \frac{e^{-u}}{u} \;
\Bigl(1 - \mathrm{Re}\bigl\{\Phi_X(u)\bigr\} \Bigr) \, \mathrm{d}u,
\end{align}
and
\begin{align}
\mathrm{Var}\bigl\{\ln(1+X^2)\bigr\} &= 2 \int_0^\infty \int_0^\infty \frac{e^{-u-v}}{uv}
\, \Bigl[ \mathrm{Re}\{\Phi_X(u+v)\} + \mathrm{Re}\{\Phi_X(u-v)\} \nonumber \\
\label{2411b}
& \hspace*{3.5cm} - 2 \, \mathrm{Re}\{\Phi_X(u)\} \, \mathrm{Re}\{\Phi_X(v)\}
\Bigr] \, \mathrm{d}u \, \mathrm{d}v.
\end{align}
\end{proposition}

As a final note, we point out that the fact that the integral representation
\eqref{ElnX} replaces the expectation of the logarithm of $X$ by the expectation
of an exponential function of $X$, has an additional interesting consequence: an
expression like $\ln(n!)$ becomes the integral of the sum of a geometric series,
which in turn is easy to express in closed form (see \cite[(2.3)--(2.4)]{Knessl98}).
Specifically,
\begin{align}
\label{lnnfactorial}
\ln(n!)&=\sum_{k=1}^n \ln k\nonumber\\
&=\sum_{k=1}^n \int_0^\infty(e^{-u}-e^{-uk}) \; \frac{\mathrm{d}u}{u}\nonumber\\
&=\int_0^\infty\left(ne^{-u}-\sum_{k=1}^ne^{-uk}\right)\frac{\mathrm{d}u}{u}\nonumber\\
&=\int_0^\infty e^{-u}\left(n-\frac{1-e^{-un}}{1-e^{-u}}\right)\frac{\mathrm{d}u}{u}.
\end{align}
Thus, for a positive integer-valued random variable, $N$, the calculation of
$\bE\{\ln N!\}$ requires merely the calculation of $\bE\{N\}$ and
the MGF, $\bE\{e^{-uN}\}$. For example, if $N$ is a Poissonian random
variable, as discussed near the end of the Introduction,
both $\bE\{N\}$ and $\bE\{e^{-uN}\}$ are easy to evaluate. This approach is a
simple, direct alternative
to the one taken in \cite{EvansB88} (see also \cite{Martinez07}), where
Malmst\'en's non--trivial formula for $\ln\Gamma(z)$ (see \cite[pp.~20--21]{EMOT81})
was invoked. (Malmst\'en's formula for $\ln\Gamma(z)$ applies to a general,
complex--valued $z$ with $\mathrm{Re}(z)>0$; in the present context, however, only
integer real values of $z$ are needed, and this allows the simplification shown in
\eqref{lnnfactorial}). The above described idea of the geometric series will also be
used in one of our application examples, in Section~\ref{subsec: empirical ent.}.

\section{Applications}
\label{appl}

In this section, we show the usefulness of the integral representation
of the logarithmic function in several problem areas in information theory.
To demonstrate the direct computability of the relevant quantities, we also
present graphs of their numerical calculation.
In some of the examples, we also demonstrate calculations of the second moments
and variances.

\subsection{Differential Entropy for Generalized
Multivariate Cauchy Densities}
\label{subsec:Cauchy}

Let $(X_1,\ldots,X_n)$ be a random vector whose probability density function is of the
form
\begin{align} \label{pdf gen. Cauchy}
f(x_1,\ldots,x_n)=\frac{C_n}{\left[1+\sum_{i=1}^n g(x_i)\right]^q}, \quad
\forall \, (x_1, \ldots, x_n) \in \reals^n,
\end{align}
for a certain non--negative function $g$ and positive constant $q$ such that
\begin{align}
\int_{\reals^n}\frac{\mathrm{d}\bx}{\left[1+\sum_{i=1}^ng(x_i)\right]^q} <
\infty.
\end{align}
We refer to this kind of density as a {\it generalized multivariate Cauchy density},
because the multivariate Cauchy density is obtained as a special case where
$g(x)=x^2$ and $q=\tfrac12 (n+1)$.
Using the Laplace transform relation,
\begin{align}
\frac{1}{s^q}=\frac{1}{\Gamma(q)}\int_0^\infty
t^{q-1}e^{-st} \, \mathrm{d}t, \quad \forall \, q>0, \; \mathrm{Re}(s) > 0,
\end{align}
$f$ can be represented as a mixture of product measures:
\begin{align}
f(x_1,\ldots,x_n)&=\frac{C_n}{\left[1+\sum_{i=1}^ng(x_i)\right]^q}\nonumber\\
\label{2911a1}
&=\frac{C_n}{\Gamma(q)} \int_0^\infty
t^{q-1}e^{-t} \, \exp\left\{-t\sum_{i=1}^ng(x_i)\right\} \, \mathrm{d}t.
\end{align}
Defining
\begin{align}
\label{eq:Z}
Z(t)\dfn\int_{-\infty}^{\infty}e^{-tg(x)} \, \mathrm{d}x, \quad \forall \, t>0,
\end{align}
we get from \eqref{2911a1},
\begin{align}
1&=\frac{C_n}{\Gamma(q)}\int_0^\infty
t^{q-1}e^{-t}\int_{\reals^n}\exp\left\{-t\sum_{i=1}^ng(x_i)\right\}
\, \mathrm{d}x_1 \, \ldots \, \mathrm{d}x_n \, \mathrm{d}t\nonumber\\
&=\frac{C_n}{\Gamma(q)} \int_0^\infty t^{q-1}e^{-t}
\left(\int_{-\infty}^{\infty} e^{-tg(x)} \, \mathrm{d}x \right)^n
\mathrm{d}t\nonumber\\
&=\frac{C_n}{\Gamma(q)}\int_0^\infty t^{q-1}e^{-t} Z^n(t) \, \mathrm{d}t,
\end{align}
and so,
\begin{align} \label{C_n}
C_n=\frac{\Gamma(q)}{\displaystyle \int_0^\infty t^{q-1}e^{-t}Z^n(t) \, \mathrm{d}t}.
\end{align}
The calculation of the differential entropy of $f$ is associated
with the evaluation of the expectation $\bE \Bigl\{\ln\bigl[1+\sum_{i=1}^ng(X_i)\bigr] \Bigr\}$.
Using \eqref{2511a},
\begin{align}
& \bE \Biggl\{ \ln\left[1+\sum_{i=1}^ng(X_i)\right] \Biggr\} \nonumber \\
\label{0812a1}
& =\int_0^\infty\frac{e^{-u}}{u}\left(1-\bE\left\{\exp\left[-u\sum_{i=1}^ng(X_i)\right]
\right\}\right)\mathrm{d}u.
\end{align}
From \eqref{2911a1} and by interchanging the integration,
\begin{align}
& \bE\left\{\exp\left[-u\sum_{i=1}^ng(X_i)\right]\right\} \nonumber \\
&= \frac{C_n}{\Gamma(q)}\int_0^\infty
t^{q-1}e^{-t}\int_{\reals^n}\exp\left\{-(t+u)\sum_{i=1}^ng(x_i)\right\}
\mathrm{d}x_1 \, \ldots \mathrm{d}x_n \, \mathrm{d}t \nonumber\\
\label{0812a2}
&=\frac{C_n}{\Gamma(q)}\int_0^\infty t^{q-1}e^{-t}Z^n(t+u) \, \mathrm{d}t.
\end{align}
In view of \eqref{2911a1}, \eqref{0812a1} and \eqref{0812a2}, the differential entropy of $(X_1,\ldots,X_n)$ is given by
\begin{align}
h(X_1,\ldots,X_n)&=
q \, \bE \Biggl\{\ln\left[1+\sum_{i=1}^ng(X_i)\right]\Biggr\}
-\ln C_n\nonumber\\[0.1cm]
&= q \int_0^\infty\frac{e^{-u}}{u}\left(1-\frac{C_n}{\Gamma(q)}\int_0^\infty
t^{q-1}e^{-t}Z^n(t+u)\mathrm{d}t\right)\mathrm{d}u-\ln C_n\nonumber\\[0.1cm]
\label{diff ent. gen. cauchy 1}
&= \frac{qC_n}{\Gamma(q)} \int_0^\infty \int_0^\infty
\frac{t^{q-1}e^{-(t+u)}}{u} \; \Bigl[Z^n(t)-Z^n(t+u)\Bigr] \, \mathrm{d}t \, \mathrm{d}u
-\ln C_n.
\end{align}

For $g(x)=|x|^\theta$, with an arbitrary $\theta > 0$, we obtain from \eqref{eq:Z} that
\begin{align} \label{eq:spec. Z}
Z(t)=\frac{2 \, \Gamma(1/\theta)}{\theta \, t^{1/\theta}}.
\end{align}
In particular, for $\theta=2$ and $q=\tfrac12 (n+1)$, we get
the multivariate Cauchy density from \eqref{pdf gen. Cauchy}.
In this case, since $\Gamma\bigl(\tfrac12\bigr) = \sqrt{\pi}$,
it follows from \eqref{eq:spec. Z} that $Z(t)=\sqrt{\frac{\pi}{t}}$ for $t>0$, and from \eqref{C_n}
\begin{align}
\label{cauchynormalization}
C_n=\frac{\Gamma\left(\frac{n+1}{2}\right)}{\pi^{n/2} \, \displaystyle \int_0^\infty
t^{(n+1)/2-1}e^{-t} \, t^{-n/2} \, \mathrm{d}t}
=\frac{\Gamma\left(\frac{n+1}{2}\right)}{\pi^{n/2} \, \Gamma\bigl(\tfrac12\bigr)}
=\frac{\Gamma\left(\frac{n+1}{2}\right)}{\pi^{(n+1)/2}}.
\end{align}
Combining \eqref{diff ent. gen. cauchy 1}, \eqref{eq:spec. Z} and
\eqref{cauchynormalization} gives
\begin{align}
h(X_1,\ldots,X_n)
&=\frac{n+1}{2\pi^{(n+1)/2}} \int_0^\infty\int_0^\infty
\frac{e^{-(t+u)}}{u\sqrt{t}}\left[1-\left(\frac{t}{t+u}\right)^{n/2}\right]
\, \mathrm{d}t \, \mathrm{d}u \nonumber\\[0.1cm]
\label{diff ent. gen. cauchy 2}
& \hspace*{0.4cm} +\frac{(n+1)\ln\pi}{2}-
\ln\Gamma\left(\frac{n+1}{2}\right).
\end{align}

\begin{figure}[h!t!b!]
\centering
\vspace*{-5cm}
\includegraphics[width=12cm]{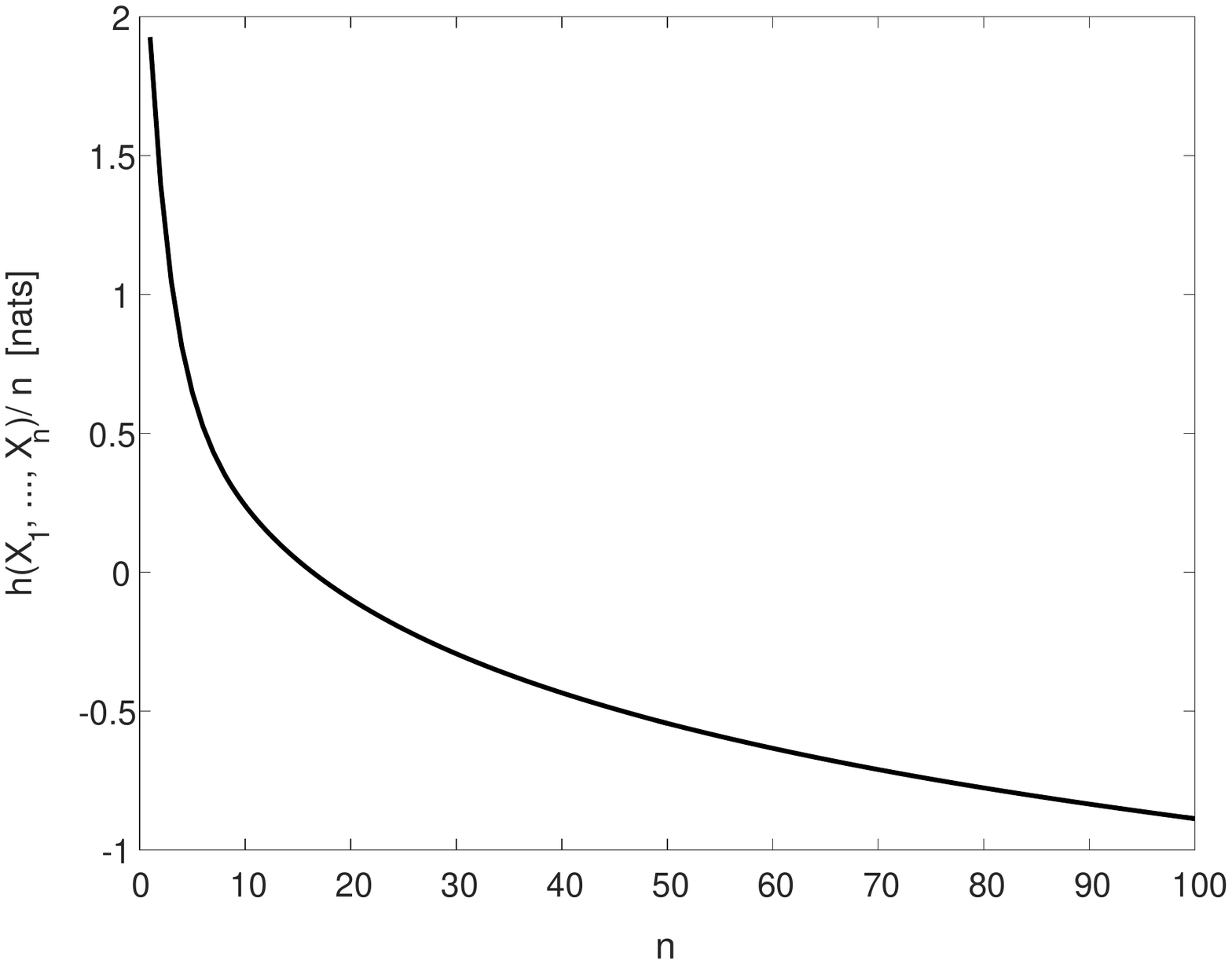}
\vspace*{-3.8cm}
\caption{The normalized differential entropy, $\tfrac1n \,
h(X_1,\ldots,X_n)$ (see \eqref{diff ent. gen. cauchy 2}), for a
multivariate Cauchy density,
$f(x_1,\ldots,x_n)=C_n/[1+\sum_{i=1}^nx_i^2]^{(n+1)/2}$, with $C_n$ in
\eqref{cauchynormalization}.}
\label{graph3}
\end{figure}
Fig.~\ref{graph3} displays the normalized differential entropy,
$\tfrac1n \, h(X_1,\ldots,X_n)$, for $1\le n\le 100$.

We believe that the interesting point, conveyed in this application example, is that
\eqref{diff ent. gen. cauchy 1} provides a kind of a ``single--letter expression'';
the $n$--dimensional integral, associated with the original expression of the
differential entropy $h(X_1, \ldots, X_n)$, is replaced by the two--dimensional
integral in \eqref{diff ent. gen. cauchy 1}, independently of~$n$.

As a final note, we mention that a lower bound on the differential entropy of
a different form of extended multivariate Cauchy distributions (cf. \cite[Eq.~(42)]{MK18})
was derived in \cite[Theorem~6]{MK18}. The latter result relies on obtaining lower
bounds on the differential entropy of random vectors whose densities are symmetric
log-concave or $\gamma$-concave (i.e., densities $f$ for which $f^\gamma$ is concave
for some $\gamma < 0$).

\subsection{Ergodic Capacity of the Fading SIMO Channel}
\label{subsec: SIMO}

Consider the SIMO channel with $L$ receive antennas and assume
that the channel transfer coefficients, $h_1,\ldots,h_L$, are
independent, zero--mean, circularly symmetric complex Gaussian
random variables with variances $\sigma_1^2,\ldots,\sigma_L^2$.
The ergodic capacity (in nats per channel use) of the SIMO channel
is given by
\begin{align} \label{capacity SIMO}
C = \bE \Biggl\{\ln\left(1+\rho \sum_{\ell=1}^L
|h_\ell|^2\right) \Biggr\} =
\bE \Biggl\{ \ln \left(1+\rho \sum_{\ell=1}^L
\bigl(f_\ell^2+g_\ell^2 \bigr)\right) \Biggr\},
\end{align}
where $f_\ell := \mathrm{Re}\{h_\ell\}$, $g_\ell := \mathrm{Im}\{h_\ell\}$,
and $\rho := \frac{P}{N_0}$ is the signal--to--noise ratio (SNR)
(see, e.g., \cite{DZWY15}, \cite{TV05} and many references therein).

Paper~\cite{DZWY15} is devoted, among other things, to the exact evaluation
of \eqref{capacity SIMO} by finding the density of the random variable
defined by $\sum_{\ell=1}^L(f_\ell^2+g_\ell^2)$, and then taking the
expectation w.r.t.\ that density. Here, we show that the integral
representation in \eqref{log int. rep.} suggests a more direct approach to
the evaluation of \eqref{capacity SIMO}.
It should also be pointed out that this approach is more flexible than the
one in \cite{DZWY15}, as the latter strongly depends on the assumption that
$\{h_i\}$ are Gaussian and statistically independent. The integral
representation approach also allows other distributions of the channel
transfer gains, as well as possible correlations between the coefficients
and/or the channel inputs. Moreover, we are also able to calculate the variance
of $\ln\left(1+\rho \sum_{\ell=1}^L |h_\ell|^2\right)$, as a measure of
the fluctuations around the mean, which is obviously related to the outage.

Specifically, in view of Proposition~\ref{Proposition 2} (see \eqref{2511a}), let
\begin{align}
X \dfn \rho \, \sum_{\ell=1}^L (f_\ell^2+g_\ell^2).
\end{align}
For all $u>0$,
\begin{align}
M_X(-u) &= \bE \biggl\{ \exp \biggl( - \rho u \,
\sum_{\ell=1}^L (f_\ell^2+g_\ell^2) \biggr) \biggr\}
\nonumber \\[0.1cm]
&= \prod_{\ell=1}^L \biggl\{ \bE\Bigl\{e^{-u\rho
f_\ell^2} \Bigr\} \; \bE\Bigl\{e^{-u\rho g_\ell^2}
\Bigr\} \biggr\} \nonumber \\[0.1cm]
\label{MGF1}
&= \prod_{\ell=1}^L \frac{1}{1+u\rho\sigma_\ell^2},
\end{align}
where \eqref{MGF1} holds since
\begin{align}
\bE\left\{e^{-u\rho f_\ell^2}\right\}
&=\bE\left\{e^{-u\rho g_\ell^2}\right\} \nonumber \\[0.1cm]
&=\int_{-\infty}^\infty \frac{\mathrm{d}w}{\sqrt{\pi\sigma_\ell^2}} \,
e^{-w^2/\sigma_\ell^2} \, e^{-u\rho w^2} \nonumber\\[0.1cm]
&=\frac{1}{\sqrt{1+u\rho\sigma_\ell^2}}.
\end{align}
From \eqref{2511a}, \eqref{capacity SIMO} and \eqref{MGF1}, the ergodic
capacity (in nats per channel use) is given by
\begin{align}
C&=\bE \left\{ \ln\left(1+\rho \sum_{\ell=1}^L \bigl(f_\ell^2+g_\ell^2 \bigr) \right) \right\} \nonumber \\[0.1cm]
&=\int_0^\infty\frac{e^{-u}}{u}\left(1-\prod_{\ell=1}^L\frac{1}{1+u\rho\sigma_\ell^2}\right) \, \mathrm{d}u \nonumber \\[0.1cm]
\label{erg. C}
&=\int_0^\infty\frac{e^{-x/\rho}}{x}\left(1-\prod_{\ell=1}^L\frac{1}{1+\sigma_\ell^2x}\right) \, \mathrm{d}x.
\end{align}
A similar approach appears in \cite[Eq.~(12)]{RajanT15}.
As for the variance, from Proposition~\ref{Proposition 2} (see \eqref{2511b}) and \eqref{MGF1},
\begin{align}
& \mathrm{Var}\left\{\ln\left(1+\rho \sum_{\ell=1}^L[f_\ell^2+g_\ell^2]\right)\right\} \nonumber \\[0.1cm]
\label{var SIMO}
&= \int_0^\infty\int_0^\infty
\frac{e^{-(x+y)/\rho}}{xy} \, \Biggl\{ \prod_{\ell=1}^L \frac{1}{1+\sigma_\ell^2(x+y)}-
\prod_{\ell=1}^L \biggl[ \frac{1}{(1+\sigma_\ell^2x)(1+\sigma_\ell^2y)} \biggr] \Biggr\}
\, \mathrm{d}x \, \mathrm{d}y.
\end{align}

A similar analysis holds for the multiple-input single-output (MISO) channel.
By partial--fraction decomposition of the expression (see the
right side of \eqref{erg. C})
$$\frac{1}{x}\left(1-\prod_{\ell=1}^L\frac{1}{1+\sigma_\ell^2x}\right),$$
the ergodic capacity $C$ can be expressed as a linear combination of integrals of the form
\begin{align}
\int_0^\infty\frac{e^{-x/\rho} \, \mathrm{d}x}{1+\sigma_\ell^2x}&=
\frac{1}{\sigma_\ell^2}
\int_0^\infty\frac{e^{-t} \, \mathrm{d}t}{t+1/(\sigma_\ell^2\rho)}\nonumber\\[0.1cm]
&= \frac{e^{1/(\sigma_\ell^2\rho)}}{\sigma_\ell^2}
\int_{1/(\sigma_\ell^2\rho)}^\infty\frac{e^{-s}}{s} \; \mathrm{d}s \nonumber\\[0.1cm]
&= \frac1{\sigma_\ell^2} \; e^{1/(\sigma_\ell^2\rho)} \;
E_1\biggl(\frac{1}{\sigma_\ell^2\rho}\biggr),
\end{align}
where $E_1(\cdot)$ is the (modified) exponential integral function, defined as
\begin{align} \label{E1 func.}
E_1(x):=\int_x^\infty\frac{e^{-s}}{s} \; \mathrm{d}s, \quad \forall \, x>0.
\end{align}
A similar representation appears also in \cite[Eq.~(7)]{DZWY15}.

Consider the example of $L=2$, $\sigma_1^2=\tfrac12$ and $\sigma_2^2=1$.
From \eqref{erg. C}, the ergodic capacity of the SIMO channel is given by
\begin{align}
C&=\int_0^\infty\frac{e^{-x/\rho}}{x}\left[1-\frac{1}{(x/2+1)(x+1)}\right] \mathrm{d}x \nonumber \\
&=\int_0^\infty \frac{e^{-x/\rho} \, (x+3) \, \mathrm{d}x}{(x+1)(x+2)} \nonumber \\
\label{erg. C 2}
&= 2 e^{1/\rho} \, E_1\biggl(\frac1{\rho}\biggr) - e^{2/\rho} \, E_1\biggl(\frac{2}{\rho}\biggr).
\end{align}
The variance in this example (see \eqref{var SIMO}) is given by
\begin{align}
& \mathrm{Var} \left\{\ln\left(1+\rho\sum_{\ell=1}^2
(f_\ell^2+g_\ell^2) \right) \right\}\nonumber\\[0.1cm]
&= \int_0^\infty \int_0^\infty\frac{e^{-(x+y)/\rho}}{xy}
\, \bigg[\frac{1}{\bigl(1+0.5(x+y)\bigr)(1+x+y)} \nonumber\\
& \hspace*{4cm} -\frac{1}{(1+0.5x)(1+0.5y)(1+x)(1+y)}\bigg]
\, \mathrm{d}x \, \mathrm{d}y \nonumber \\[0.1cm]
\label{var SIMO 2}
&=\int_0^\infty \int_0^\infty
\frac{e^{-(x+y)/\rho} \, (2xy+6x+6y+10) \, \mathrm{d}x \,
\mathrm{d}y}{(x+1)(y+1)(x+2)(y+2)(x+y+1)(x+y+2)}.
\end{align}
Fig.~\ref{graph2a} depicts the ergodic capacity $C$ as a function of the SNR,
$\rho$, in dB (see \eqref{erg. C 2}, and divide by $\ln 2$ for conversion to bits per channel use).
The same example exactly appears in the lower graph of Fig.~1 in \cite{DZWY15}. The variance
appears in Fig.~\ref{graph2b} (see \eqref{var SIMO 2}, and similarly divide by $\ln^2 2$).

\begin{figure}[h!t!b!]
\centering
\vspace*{-4.4cm}
\includegraphics[width=11.5cm]{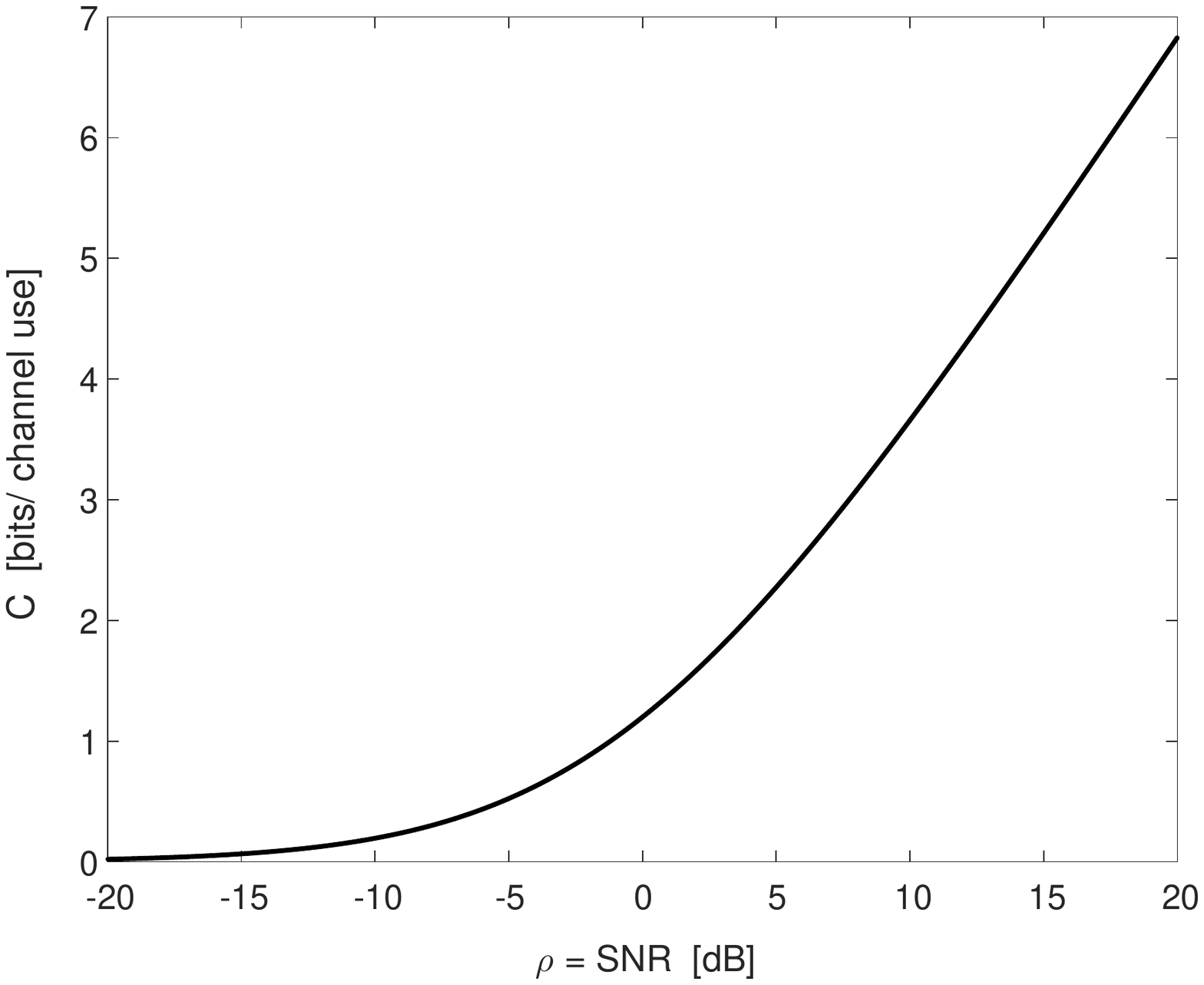}
\vspace*{-3.7cm}
\caption{The ergodic capacity $C$ (in bits per channel use) of the SIMO
channel as a function of $\rho = \mathrm{SNR}$ (in dB) for $L=2$ receive antennas,
with noise variances $\sigma_1^2=\tfrac12$ and $\sigma_2^2=1$.}
\label{graph2a}
\end{figure}

\begin{figure}[h!t!b!]
\centering
\vspace*{-4.4cm}
\includegraphics[width=11.5cm]{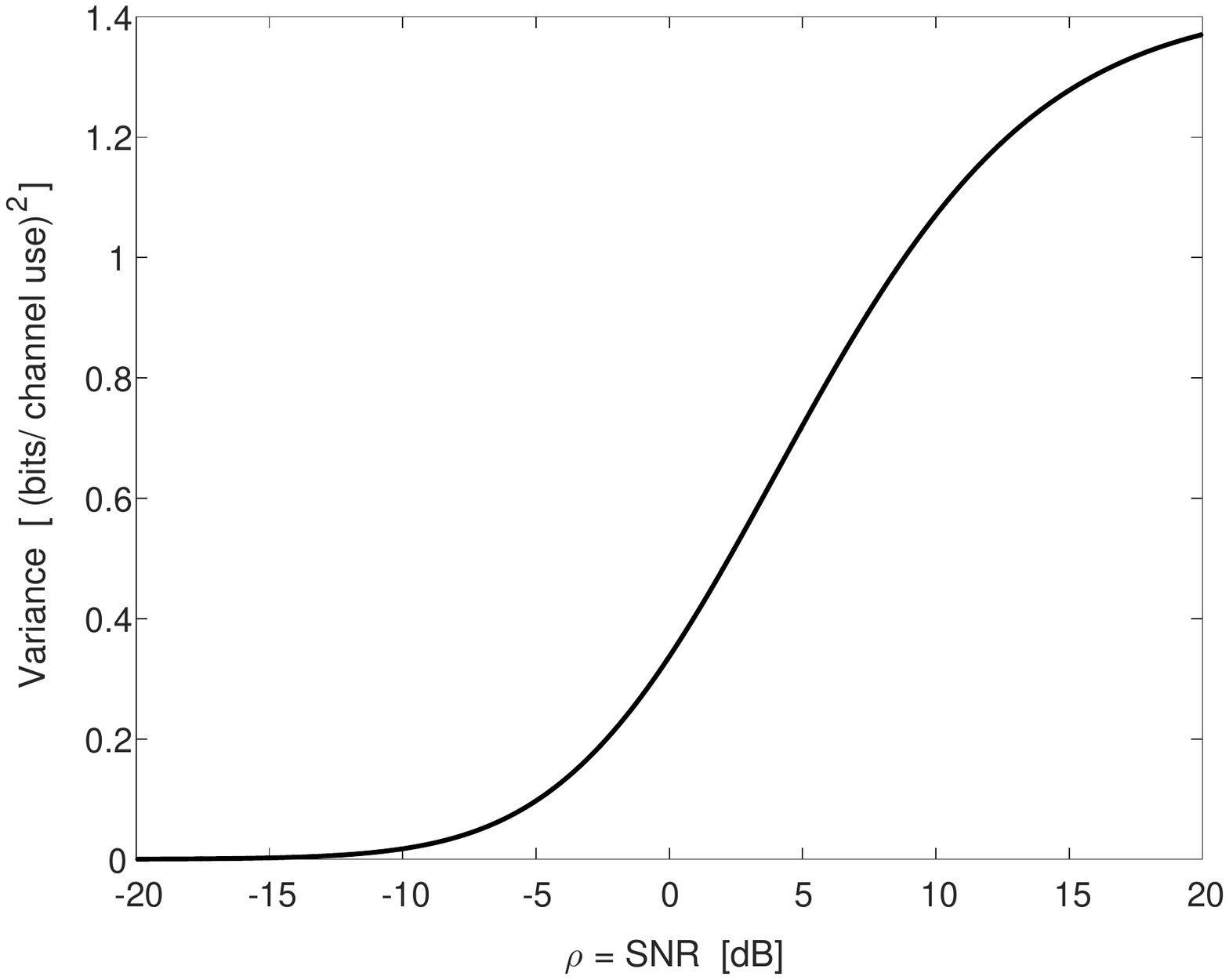}
\vspace*{-3.7cm}
\caption{The variance of $\ln(1+\rho\sum_{\ell=1}^L|h_\ell|^2)$ (in
$[\mbox{bits-per-channel-
use}]^2$) of the SIMO channel
as a function of $\rho = \mathrm{SNR}$ (in dB) for $L=2$ receive antennas,
with noise variances $\sigma_1^2=\tfrac12$ and $\sigma_2^2=1$.}
\label{graph2b}
\end{figure}

\subsection{Universal Source Coding for Binary Arbitrarily Varying Sources}
\label{subsec: universal coding}

Consider a source coding setting, where there are $n$ binary DMS's,
and let $x_i \in [0,1]$ denote the Bernoulli parameter of source no.~$i
\in \{1, \ldots, n\}$. Assume that a hidden memoryless switch selects
uniformly at random one of these sources, and the data is then emitted by
the selected source. Since it is unknown a-priori which source is selected
at each instant, a universal lossless source encoder (e.g., a Shannon or
Huffman code) is designed to match a binary DMS whose Bernoulli parameter
is given by $\tfrac1n \overset{n}{\underset{i=1}{\sum}} x_i$. Neglecting
integer length constraints, the average redundancy in the compression rate
(measured in nats per symbol), due to the unknown realization of the hidden
switch, is about
\begin{align}
\label{compression1}
R_n = h_{\mathrm{b}}\biggl( \tfrac1n \sum_{i=1}^n x_i \biggr)
- \tfrac1n \sum_{i=1}^n h_{\mathrm{b}}(x_i),
\end{align}
where $h_{\mathrm{b}} \colon [0,1] \to [0, \ln 2]$ is the binary entropy function
(defined to the base~$e$), and the redundancy is given in nats per source symbol.
Now, let us assume that the Bernoulli parameters of the $n$ sources are
i.i.d. random variables, $X_1, \ldots, X_n$, all having the same density
as that of some generic random variable $X$,
whose support is the interval $[0,1]$.
We wish to evaluate the expected value of the above defined redundancy,
under the assumption that the realizations of $X_1,\ldots,X_n$ are known.
We are then facing the need to evaluate
\begin{align}
\label{compression2}
\overline{R}_n=\bE\biggl\{ h_{\mathrm{b}}\biggl( \tfrac1n \sum_{i=1}^n X_i \biggr) \biggr\}
- \bE\{h_{\mathrm{b}}(X)\}.
\end{align}
We now express the first and second terms on the right--hand side of \eqref{compression2}
as a function of the MGF of $X$.

In view of \eqref{log int. rep.}, the binary entropy function $h_{\mathrm{b}}$ admits the
integral representation
\begin{align}
\label{bin. ent. 1}
h_{\mathrm{b}}(x) = \int_0^{\infty} \frac1{u} \, \Bigl[ x e^{-ux} + (1-x) e^{-u(1-x)} -
e^{-u} \Bigr] \, \mathrm{d}u, \quad \forall \, x \in [0,1],
\end{align}
which implies that
\begin{align}  \label{bin. ent. 2}
\bE\{ h_{\mathrm{b}}(X) \} &= \int_0^{\infty} \frac1{u} \, \Bigl[ \bE\bigl\{X
e^{-uX}\bigr\} + \bE\bigl\{(1-X) e^{-u(1-X)}\bigr\} - e^{-u} \Bigr] \,
\mathrm{d}u.
\end{align}
The expectations on the right--hand side of \eqref{bin. ent. 2} can be
expressed as functionals of the MGF of $X$,
$M_X(\nu)=\bE\{e^{\nu X}\}$, and its derivative, for $\nu < 0$.
For all $u \in \reals$,
\begin{align}
\label{bin. ent. 3}
& \bE\bigl\{X e^{-uX}\bigr\} = M_X'(-u),
\end{align}
and
\begin{align}
\bE \bigl\{(1-X) e^{-u(1-X)} \bigr\} &= M'_{1-X}(-u) \nonumber \\
&= \frac{\mathrm{d}}{\mathrm{d}s} \Bigl\{e^s \, M_X(-s) \Bigr\}\Bigl|_{s=-u}
\nonumber \\
\label{bin. ent. 4}
&= e^{-u} \Bigl[ M_X(u) - M_X'(u) \Bigr].
\end{align}
On substituting \eqref{bin. ent. 3} and \eqref{bin. ent. 4} into
\eqref{bin. ent. 2}, we readily obtain
\begin{align}
\label{bin. ent. 5}
\bE\{ h_{\mathrm{b}}(X) \} &= \int_0^{\infty} \frac1{u} \, \Bigl\{M_X'(-u) + \bigl[M_X(u) -
M_X'(u) -1\bigr] e^{-u} \Bigr\} \, \mathrm{d}u.
\end{align}
Define $Y_n \dfn \tfrac1n \overset{n}{\underset{i=1}{\sum}} X_i$. Then,
\begin{align}
\label{bin. ent. 6}
M_{Y_n}(u) = M_X^n\bigl(\tfrac{u}{n}\bigr), \quad \forall \, u \in \reals,
\end{align}
which yields, in view of \eqref{bin. ent. 5}, \eqref{bin. ent. 6} and the
change of integration variable, $t=\frac{u}{n}$, the following:
\begin{align}
&\bE\biggl\{ h_{\mathrm{b}}\biggl( \tfrac1n \sum_{i=1}^n X_i \biggr) \biggr\} \nonumber \\
&=\bE\{ h_{\mathrm{b}}(Y_n) \} \nonumber \\[0.1cm]
&= \int_0^{\infty} \frac1{u} \, \Bigl\{M_{Y_n}'(-u) + \bigl[M_{Y_n}(u) -
M_{Y_n}'(u) -1\bigr] e^{-u} \Bigr\} \, \mathrm{d}u \nonumber \\[0.1cm]
\label{bin. ent. 7}
&= \int_0^{\infty} \frac1t \, \Bigl\{ M_{X}^{n-1}(-t) \, M_{X}'(-t) +
\bigl[ M_{X}^n(t)- M_{X}^{n-1}(t) \, M_{X}'(t) - 1 \bigr] \, e^{-nt}
\Bigr\} \, \mathrm{d}t.
\end{align}
Similarly as in the application example of the differential entropy
in Section~\ref{subsec:Cauchy}, here too, we pass from an $n$-dimensional
integral to a one-dimensional integral. In general, similar calculations
can be carried out for higher integer moments, thus passing from
$n$-dimensional integration for a moment of order $s$ to an $s$-dimensional
integral, independently of $n$.

For example, if $X_1, \ldots, X_n$ are i.i.d. and uniformly distributed on
[0,1], then the MGF of a generic random variable $X$ distributed like all
$\{X_i\}$ is given by
\begin{align}
M_X(t) =
\begin{dcases}
\frac{e^t-1}{t}, \quad & t \neq 0, \\
\hspace*{0.4cm} 1, \quad &t=0.
\end{dcases}
\end{align}
From \eqref{bin. ent. 7}, it can be verified numerically that
$\bE\biggl\{ h_{\mathrm{b}}\biggl( \tfrac1n \overset{n}{\underset{i=1}{\sum}}
X_i \biggr) \biggr\}$ is monotonically increasing in~$n$,
being equal (in nats) to $\tfrac12$, 0.602, 0.634, 0.650, 0.659 for
$n=1,\ldots,5$, respectively, with the limit $h_{\mathrm{b}}\bigl(\tfrac12\bigr) =
\ln 2 \approx 0.693$ as we let $n \to \infty$ (this is expected
by the law of large numbers).

\subsection{Moments of the Empirical Entropy and the Redundancy of K--T
Universal Source Coding}
\label{subsec: empirical ent.}

Consider a stationary, discrete memoryless source (DMS), $P$, with
a finite alphabet $\calX$ of size $|\calX|$ and letter probabilities
$\{P(x),~x\in\calX\}$. Let $(X_1,\ldots,X_n)$ be an $n$--vector emitted
from $P$, and let $\{\hP(x),~x\in\calX\}$ be the empirical distribution
associated with $(X_1,\ldots,X_n)$, that is, $\hP(x)=\frac{n(x)}{n}$,
for all $x \in\calX$, where $n(x)$ is the number of occurrences of the
letter $x$ in $(X_1,\ldots,X_n)$.

It is well known that in many universal lossless source codes for the
class of memoryless sources, the dominant term of the length function
for encoding $(X_1,\ldots,X_n)$ is $n\hH$, where $\hH$ is the empirical
entropy,
\begin{align}
\hat{H}=-\sum_x \hP(x)\ln\hP(x).
\end{align}
For code--length performance analysis (as well as for entropy estimation
per se), there is therefore interest in calculating the expected value
$\bE\{\hH\}$ as well as $\mathrm{Var}\{\hH\}$.
Another motivation comes from the quest for estimating the entropy as an
objective on its own right, and then the expectation and the variance suffice
for the calculation of the mean square error of the estimate, $\hH$.
Most of the results that are available in the literature, in this context,
concern the asymptotic behavior for large $n$ as well as bounds (see, e.g.,
\cite{BRY98}, \cite{Blumer87}, \cite{ClarkeB90}, \cite{CB94}, \cite{Davisson73},
\cite{Davisson83}, \cite{DMPW81}, \cite{KT81}, \cite{MF98}, \cite{Rissanen83},
\cite{Rissanen84}, \cite{Rissanen96}, \cite{Shtarkov87}, \cite{WRF95},
\cite{XB97}, as well as many other related references therein).
The integral representation of the logarithm in \eqref{log int. rep.}, on
the other hand, allows exact calculations of the expectation and the variance.
The expected value of the empirical entropy is given by
\begin{align}
\bE\{\hat{H}\}&= -\sum_x\bE\{\hP(x)\ln\hP(x)\}\nonumber\\
&=\sum_x\bE\left\{\int_0^\infty\frac{\mathrm{d}u}{u}
\left[\hP(x)e^{-u\hP(x)}-\hP(x)e^{-u}\right]\right\}\\
\label{EH_0212}
&=\int_0^\infty\frac{\mathrm{d}u}{u}
\left[\sum_x\bE\{\hP(x)e^{-u\hP(x)}\}-e^{-u}\right].
\end{align}
For convenience, let us define the function
$\phi_n \colon \mathcal{X} \times \reals \to (0,\infty)$ as
\begin{align} \label{phi}
\phi_n(x,t)\dfn \bE\bigl\{e^{t\hP(x)}\bigr\}=\Bigl[1-P(x)+P(x)e^{t/n}\Bigr]^n,
\end{align}
which yields,
\begin{align}
\label{diff1 phi}
& \bE\bigl\{\hP(x)e^{-u\hP(x)}\bigr\}=\phi_n'(x,-u), \\
\label{diff2 phi}
& \bE\bigl\{\hP^2(x)e^{-u\hP(x)}\bigr\}=\phi_n''(x,-u),
\end{align}
where $\phi_n'$ and $\phi_n''$ are the first and second order
derivatives of $\phi_n$ w.r.t.\ $t$, respectively. From
\eqref{EH_0212} and \eqref{diff1 phi}, it follows that
\begin{align}
\label{int1a}
\bE\{\hat{H}\}
&=\int_0^\infty\frac{\mathrm{d}u}{u}\left(\sum_x\phi_n'(x,-u)
-e^{-u}\right) \\
\label{int1b}
&=\int_0^\infty \frac{\mathrm{d}u}{u}\left(e^{-u}\sum_x P(x)
\left[1-P(x)(1-e^{-u})\right]^{n-1}-e^{-nu}\right)
\end{align}
where the integration variable in \eqref{int1b} was changed
using a simple scaling by $n$.

Before proceeding with the calculation of the variance of $\hat{H}$,
let us first compare the
integral representation in \eqref{int1b} to the alternative sum,
obtained by a direct, straightforward calculation of the expected
value of the empirical entropy. A straightforward calculation gives
\begin{align}
\bE\{\hH\} &=
\sum_x\sum_{k=0}^n
\binom{n}{k} \, P^k(x) \, [1-P(x)]^{n-k}
\cdot \frac{k}{n} \cdot \ln\frac{n}{k}\\
&= \sum_x \sum_{k=1}^n \binom{n-1}{k-1} P^{k}(x)
[1-P(x)]^{n-k} \cdot \ln \frac{n}{k}. \label{from log rep.}
\end{align}
We next compare the computational complexity of implementing \eqref{int1b}
to that of \eqref{from log rep.}. For large $n$, in order to avoid numerical
problems in computing \eqref{from log rep.} by standard software, one may
use the {\tt{Gammaln}} function in Matlab/Excel or the {\tt LogGamma} in
Mathematica (a built-in function for calculating the natural logarithm of
the Gamma function) to obtain that
\begin{align}
\binom{n-1}{k-1} P^{k}(x) [1-P(x)]^{n-k}
& = \exp \bigg\{\mathrm{Gammaln}(n)-\mathrm{Gammaln}(k)-
\mathrm{Gammaln}(n-k+1) \nonumber \\
& \hspace*{1.5cm} + k\ln P(x) +(n-k) \ln \bigl(1-P(x)\bigr) \bigg\}.
\end{align}
The right--hand side of \eqref{int1b} is the sum of $|\mathcal{X}|$
integrals, and the computational complexity of each integral depends
on neither $n$, nor $|\mathcal{X}|$. Hence, the computational complexity
of the right--hand side of \eqref{int1b} scales {\em linearly} with
$|\mathcal{X}|$. On the other hand, the double sum on the right--hand
side of \eqref{from log rep.} consists of $n\cdot|\mathcal{X}|$ terms.
Let $\alpha \dfn \frac{n}{|\mathcal{X}|}$ be fixed, which is expected
to be large ($\alpha \gg 1$) if a good estimate of the entropy is sought.
The computational complexity of the double sum on the right--hand side
of \eqref{from log rep.} grows like $\alpha \, |\mathcal{X}|^2$, which
scales {\em quadratically} in $|\mathcal{X}|$.
Hence, for a DMS with a large alphabet, or when $n \gg |X|$, there is
a significant computational reduction by evaluating \eqref{int1b} in
comparison to the right--hand side of \eqref{from log rep.}.

We next move on to calculate the variance of $\hH$.
\begin{align}
\mathrm{Var}\{\hat{H}\}
&=\bE\{\hH^2\}-\bE^2\{\hH\}\\
\label{201119a1}
&=\sum_{x,x'}\bE\{
\hP(x)\ln\hP(x) \cdot \hP(x')\ln\hP(x')\}-\bE^2\{\hH\}.
\end{align}
The second term on the right--hand side of \eqref{201119a1} has already been
calculated. For the first term, let us define, for $x'\ne x$,
\begin{align}
\psi_n(x,x',s,t)
&\dfn \bE\{\exp\{s\hP(x)+t\hP(x')\} \\
&=\sum_{\{(k,\ell):~k+\ell\le n\}} \biggr\{ \frac{n!}{k! \,
\ell! \, (n-k-\ell)!} \cdot P^{k}(x) \, P^{\ell}(x') \nonumber\\
& \hspace*{3.2cm} \cdot \bigl[1-P(x)-P(x')\bigr]^{n-k-\ell} \,
e^{sk/n+t\ell/n} \biggr\} \\
&=\sum_{\{(k,\ell):~k+\ell\le n\}} \biggl\{ \frac{n!}{k! \, \ell! \,
(n-k-\ell)!} \cdot \bigl[P(x) \, e^{s/n}\bigr]^{k} \,
\bigl[P(x') \, e^{t/n}\bigr]^{\ell} \nonumber\\
& \hspace*{3.2cm} \cdot \bigl[1-P(x)-P(x')\bigr]^{n-k-\ell} \biggr\} \\
\label{eq: psi}
&=\left[1-P(x) \, (1-e^{s/n})-P(x') \, (1-e^{t/n})\right]^n.
\end{align}
Observe that
\begin{align}
\bE\{\hP(x)\hP(x')\exp\{-u\hP(x)-v\hP(x')\}
&=\frac{\partial^2\psi_n(x,x',s,t)}{\partial
s \, \partial t}\bigg|_{s=-u, \, t=-v} \\
\label{diff2 psi}
& := \psi_n''(x,x',-u,-v).
\end{align}
For $x\ne x'$, we have
\begin{align}
&\bE\{\hP(x)\ln\hP(x) \cdot \hP(x')\ln\hP(x')\} \nonumber \\[0.05cm]
&= \bE\left\{\hP(x)\hP(x')\int_0^\infty \int_0^\infty \frac{\mathrm{d}u \,
\mathrm{d}v}{uv}\cdot \bigl(e^{-u}-e^{-u\hP(x)} \bigr)\cdot
\bigl(e^{-v}-e^{-v\hP(x')} \bigr) \right\} \label{201119a2} \\[0.1cm]
&=\int_0^\infty\int_0^\infty \frac{\mathrm{d}u \, \mathrm{d}v}{uv} \;
\bigg[e^{-u-v} \bE\bigl\{\hP(x)\hP(x')\bigr\}
-e^{-v} \bE\bigl\{\hP(x)\hP(x')e^{-u\hP(x)}\bigr\} \nonumber\\[0.1cm]
& \hspace*{3cm} -e^{-u} \bE\bigl\{\hP(x)\hP(x')e^{-v\hP(x')}\bigr\}
+\bE\bigl\{\hP(x) \hP(x') e^{-u\hP(x)-v\hP(x')}\bigr\}\bigg] \label{201119a3} \\[0.1cm]
&=\int_0^\infty \int_0^\infty\frac{\mathrm{d}u \, \mathrm{d}v}{uv} \;
\bigg[e^{-u-v}\psi_n''(x,x',0,0)-e^{-v}\psi_n''(x,x',-u,0) \nonumber \\[0.1cm]
& \hspace*{3cm} -e^{-u}\psi_n''(x,x',0,-v)+\psi_n''(x,x',-u,-v)\bigg],
\label{201119a4}
\end{align}
and for $x=x'$,
\begin{align}
& \bE\{[\hP(x)\ln\hP(x)]^2\} \nonumber \\[0.1cm]
&= \bE\left\{\hP^2(x)\int_0^\infty\int_0^\infty\frac{\mathrm{d}u \,
\mathrm{d}v}{uv}\cdot
\bigl[e^{-u}-e^{-u\hP(x)}\bigr]\cdot
\bigl[e^{-v}-e^{-v\hP(x)}\bigr]\right\} \label{201119a5} \\[0.1cm]
&= \int_0^\infty\int_0^\infty\frac{\mathrm{d}u \, \mathrm{d}v}{uv} \;
\bigg[e^{-u-v}\bE\bigl\{\hP^2(x)\bigr\}-e^{-v}\bE\bigl\{\hP^2(x)e^{-u\hP(x)}\bigr\}
\nonumber\\[0.1cm]
& \hspace*{3.2cm} -e^{-u}\bE\bigl\{\hP^2(x)e^{-v\hP(x)}\bigr\}
+\bE\bigl\{\hP^2(x)e^{-(u+v)\hP(x)}\bigr\}\bigg] \label{201119a6} \\[0.1cm]
& =\int_0^\infty\int_0^\infty\frac{\mathrm{d}u \, \mathrm{d}v}{uv} \;
\bigg[e^{-u-v}\phi_n''(x,0)-e^{-v}\phi_n''(x,-u) \nonumber \\[0.1cm]
& \hspace*{3.2cm} -e^{-u}\phi_n''(x,-v)+\phi_n''(x,-u-v)\bigg].
\label{201119a7}
\end{align}
Therefore,
\begin{align}
& \mathrm{Var}\{\hH\} \nonumber \\[0.1cm]
&= \sum_x\int_0^\infty\int_0^\infty\frac{\mathrm{d}u \, \mathrm{d}v}{uv} \;
\bigg[e^{-u-v}\phi_n''(x,0)-e^{-v}\phi_n''(x,-u)-
e^{-u}\phi_n''(x,-v)+\phi_n''(x,-u-v)\bigg] \nonumber\\[0.1cm]
&\hspace*{0.4cm} +\sum_{x'\ne x}\int_0^\infty\int_0^\infty\frac{\mathrm{d}u \,
\mathrm{d}v}{uv} \;
\bigg[e^{-u-v}\psi_n''(x,x',0,0)-e^{-v}\psi_n''(x,x',-u,0) \nonumber\\[0.1cm]
& \hspace*{4.4cm} -e^{-u}\psi_n''(x,x',0,-v) +
\psi_n''(x,x',-u,-v)\bigg] -\bE^2\{\hH\}. \label{201119a10}
\end{align}
Defining (see \eqref{diff2 phi} and \eqref{diff2 psi})
\begin{align}
Z(r,s,t) \dfn \sum_x\phi_n''(x,r)+\sum_{x'\ne x}\psi_n''(x,x',s,t),
\end{align}
we have
\begin{align}
\mathrm{Var}\{\hH\}
&=\int_0^\infty \int_0^\infty \frac{\mathrm{d}u \, \mathrm{d}v}{uv} \,
\Bigl[e^{-u-v}Z(0,0,0)-e^{-v}Z(-u,-u,0) \nonumber\\[0.1cm]
& \hspace*{3.5cm} -e^{-u}Z(-v,0,-v)+Z(-u-v,-u,-v)\Bigr]-\bE^2\{\hH\}.
\end{align}

To obtain numerical results, it would be convenient to
particularize now the analysis to the binary symmetric source (BSS). From
\eqref{int1b},
\begin{align} \label{int BSS}
\bE\{\hH\}=\int_0^\infty \frac{\mathrm{d}u}{u} \left[e^{-u} \, \Bigl(\frac{
1+e^{-u}}{2}\Bigr)^{n-1}-e^{-un}\right].
\end{align}
For the variance, it follows from \eqref{eq: psi} that
for $x \neq x'$ with $x,x' \in \{0,1\}$ and $s,t \in \reals$,
\begin{align}
\psi_n(x,x',s,t)
&= \biggl(\frac{e^{s/n} + e^{t/n}}{2}\biggr)^n, \\[0.1cm]
\psi_n''(x,x',s,t)
&= \frac{\partial^2\psi_n(x,x',s,t)}{\partial s \, \partial t}
= \tfrac14 \Bigl(1-\tfrac1n\Bigr) \biggl(\frac{e^{s/n} +
e^{t/n}}{2}\biggr)^{n-2} e^{(s+t)/n},
\end{align}
and, from \eqref{201119a2}--\eqref{201119a4}, for $x \neq x'$
\begin{align}
&\bE\{\hP(x)\ln\hP(x) \cdot \hP(x')\ln\hP(x')\} \nonumber \\[0.1cm]
&= \tfrac14 \Bigl(1-\tfrac1n\Bigr) \int_0^\infty \int_0^\infty \frac{\mathrm{d}u
\, \mathrm{d}v}{u v} \, \biggl[ e^{-u-v}
- e^{-\bigl(u/n+v\bigr)} \biggl(\frac{1 +
e^{-u/n}}{2}\biggr)^{n-2} \nonumber \\[0.1cm]
&\hspace*{5cm} - e^{-(u+v/n)} \biggl(\frac{1 +
e^{-v/n}}{2}\biggr)^{n-2} \nonumber \\[0.1cm]
&\hspace*{5cm}
+ e^{-(u+v)/n} \biggl(\frac{e^{-u/n} +
e^{-v/n}}{2}\biggr)^{n-2}\biggr] \label{201119a8}.
\end{align}
From \eqref{phi}, for $x \in \{0,1\}$ and $t \in \reals$,
\begin{align}
& \phi_n(x,t) = \biggl(\frac{1 + e^{t/n}}{2}\biggr)^n, \\[0.1cm]
& \phi_n''(x,t) = \frac{\partial^2 \phi_n(x,t)}{\partial t^2}
= \frac{e^{t/n}}{4n} \biggl(\frac{1 +
e^{t/n}}{2}\biggr)^{n-2} \Bigl(1 + n e^{t/n} \Bigr),
\end{align}
and, from \eqref{201119a5}--\eqref{201119a7}, for $x \in \{0,1\}$,
\begin{align}
&\bE\{[\hP(x)\ln\hP(x)]^2\} \nonumber \\[0.1cm]
&= \tfrac1{4n} \int_0^\infty \int_0^\infty \frac{\mathrm{d}u \, \mathrm{d}v}{u v}
\, \biggl\{ (n+1) e^{-u-v} - e^{-\bigl(u/n+v\bigr)} \biggl(\frac{
1 + e^{-u/n}}{2}\biggr)^{n-2} \, \Bigl(1 + n e^{-u/n}
\Bigr) \nonumber \\[0.1cm]
& \hspace*{4cm} - e^{-\bigl(u+v/n\bigr)} \biggl(\frac{1 +
e^{-v/n}}{2}\biggr)^{n-2} \, \Bigl(1 + n e^{-v/n} \Bigr)
\nonumber \\[0.1cm]
& \hspace*{4cm} + e^{-(u+v)/n} \biggl(\frac{1 +
e^{-(u+v)/n}}{2}\biggr)^{n-2} \, \Bigl(1 + n e^{-(u+v)/n} \Bigr)
\biggr\} \label{201119a9}.
\end{align}
Combining \eqref{201119a10}, \eqref{201119a8} and \eqref{201119a9} gives
the following closed--form expression for the variance of the empirical
entropy:
\begin{align}
& \mathrm{Var}\{\hH\} \nonumber \\
&= \tfrac12 \Bigl(1+\tfrac1n\Bigr) \int_0^\infty
\int_0^\infty \frac{\mathrm{d}u \, \mathrm{d}v}{u v} \,
\biggl[ e^{-(u+v)} - e^{-v} f_n\Bigl(\frac{u}{n}\Bigr) - e^{-u}
f_n\Bigl(\frac{v}{n}\Bigr)
+ f_n\Bigl(\frac{u+v}{n}\Bigr) \biggr] \nonumber \\[0.1cm]
& \hspace*{0.5cm} + \tfrac12 \Bigl(1-\tfrac1n\Bigr) \int_0^\infty
\int_0^\infty \frac{\mathrm{d}u \, \mathrm{d}v}{u v} \, \biggl[ e^{-(u+v)}
-e^{-v} g_n\Bigl(\frac{u}{n}, 0\Bigr) - e^{-u} g_n\Bigl(0, \frac{v}{n}\Bigr) +
g_n\Bigl(\frac{u}{n}, \frac{v}{n}\Bigr) \biggr]
\nonumber \\[0.1cm]
& \hspace*{0.5cm} - \left\{ \int_0^\infty \frac{\mathrm{d}u}{u} \left[e^{-u} \,
\Bigl(\frac{1+e^{-u}}{2}\Bigr)^{n-1}-e^{-un}\right] \right\}^2,
\label{int2 BSS}
\end{align}
where
\begin{align}
& f_n(s) \dfn e^{-s} \, \biggl(\frac{1+e^{-s}}{2}\biggr)^{n-2} \,
\frac{1+ne^{-s}}{n+1}, \\[0.1cm]
& g_n(s,t) = e^{-s-t} \, \biggl(\frac{e^{-s}+e^{-t}}{2}\biggr)^{n-2}.
\end{align}

\begin{figure}[h!t!b!]
\centering
\vspace*{-4.5cm}
\includegraphics[width=11.5cm]{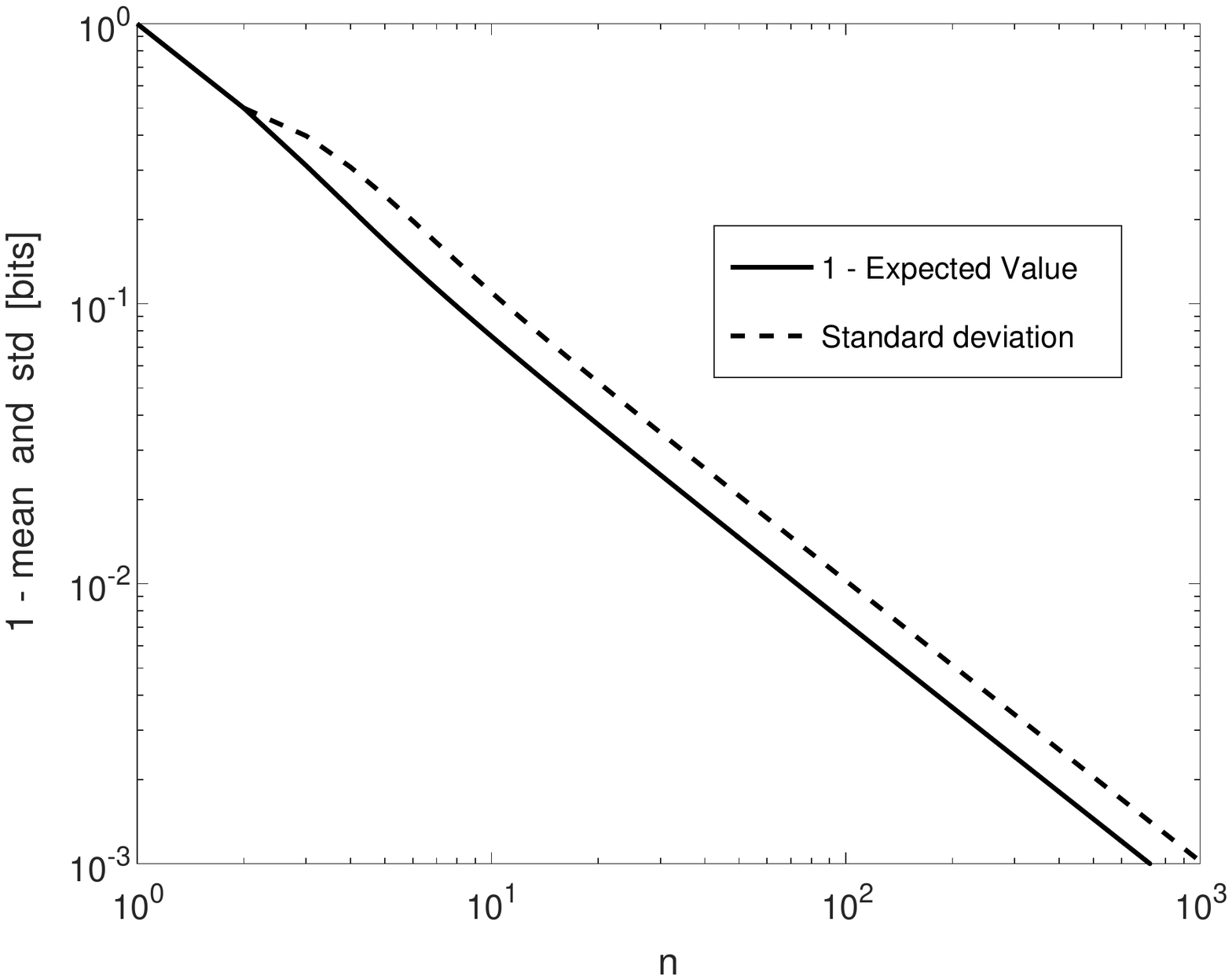}
\vspace*{-3.8cm}
\caption{$1-\bE\{\hH\}$ and $\mathrm{std}(\hH)$ for a BSS (in bits per source
symbol) as a function of $n$.}
\label{expempent}
\end{figure}
For the BSS,
$\ln 2 -\bE\{\hH\} = \bE\{D(\hP \| P)\}$ and the standard deviation of
$\hH$ both decay at the rate of $\frac1{n}$ as $n$ grows without bound,
according to Fig.~\ref{expempent}.
This asymptotic behavior of $\bE\{D(\hP\|P)\}$ is supported by the
well-known result \cite{Wald43} (see also \cite[Section~3.C]{ClarkeB90}
and references therein) that for the class of discrete memoryless
sources $\{P\}$ with a given finite alphabet $\calX$,
\begin{align} \label{Wald}
& \ln \frac{\hat{P}(X_1,\ldots,X_n)}{P(X_1,\ldots,X_n)}
\to \tfrac12 \, \chi_d^2,
\end{align}
in law, where $\chi_d^2$ is a chi-squared random variable with $d$ degrees of
freedom. The left--hand side of \eqref{Wald} can be rewritten as
\begin{align} \label{Wald2}
\ln\left(\frac{\exp\{-n\hH\}}{\exp\{-n\hH-nD(\hat{P}\|
P)\}}\right) = nD(\hat{P}\|P),
\end{align}
and so, $\bE\{D(\hP \| P)\}$ decays like
$\frac{d}{2n}$,
which is equal to $\frac1{2n}$ for the BSS. In Fig.~\ref{expempent},
the base of the logarithm is~2, and therefore,
$\bE\{D(\hP\|P)\} = 1 - \bE\{\hH\}$ decays like $\frac{\log_2
\mathrm{e}}{2n} \approx \frac{0.7213}{n}$.
It can be verified numerically that $1 - \bE\{\hH\}$ (in bits) is equal
to $7.25 \cdot 10^{-3}$ and $7.217 \cdot 10^{-4}$ for $n=100$ and $n=1000$,
respectively (see Fig.~\ref{expempent}), which confirms \eqref{Wald}
and \eqref{Wald2}. Furthermore, the exact result here for the standard
deviation, which decays like $\frac{1}{n}$, scales similarly to the concentration
inequality in \cite[(9)]{Mardia19}.

We conclude this subsection by exploring a quantity related to the empirical entropy,
which is the expected code length associated with the universal lossless source code due
to Krichevsky and Trofimov \cite{KT81}. In a nutshell, this is a predictive universal
code, which at each time instant $t$, sequentially assigns probabilities to the next
symbol according to (a biased version of) the empirical distribution pertaining to the
data seen thus far, $x_1,\ldots,x_t$. Specifically, consider the code--length function
(in nats),
\begin{align}  \label{KT81a}
L(x^n) = -\sum_{t=0}^{n-1}\ln Q(x_{t+1}|x^t),
\end{align}
where
\begin{align}  \label{KT81b}
Q(x_{t+1}=x|x_1,\ldots,x_t) = \frac{N_t(x)+s}{t+s|\calX|},
\end{align}
$N_t(x)$ is the number of occurrences of the symbol $x\in\calX$ in
$(x_1,\ldots,x_t)$, and $s > 0$ is a fixed bias parameter needed for the
initial coding distribution ($t=0$).

We now calculate the redundancy of this universal code,
\begin{align} \label{redun.}
R_n = \frac{\bE\{L(X^n)\}}{n} - H,
\end{align}
where $H$ is the entropy of the underlying source.
From \eqref{KT81a}, \eqref{KT81b} and \eqref{redun.}, we can represent $R_n$ as follows:
\begin{align}  \label{KT81c}
R_n = \frac{1}{n}\sum_{t=0}^{n-1}\bE\left\{\ln\frac{(t+s|\calX|)P(X_{t+1})}{N_t(X_{t+1})+s}\right\}.
\end{align}
The expectation on the right--hand side of \eqref{KT81c} satisfies
\begin{align}
&\bE\left\{\ln\frac{(t+s|\calX|)P(X_{t+1})}{N_t(X_{t+1})+s}\right\} \nonumber\\[0.15cm]
&=\sum_x P(x)\bE\left\{\ln\frac{(t+s|\calX|)P(x)}{N_t(x)+s}\right\}\nonumber\\[0.1cm]
&=\int_0^\infty\left[e^{-us}\sum_x P(x)\bE\{e^{-uN_t(x)}\}-
\sum_x P(x)e^{-u(s|\calX|+t)P(x)}\right] \, \frac{\mathrm{d}u}{u}\nonumber\\[0.15cm]
\label{KT81d}
&=\int_0^\infty\left[e^{-us}\sum_x P(x)[1-P(x)(1-e^{-u})]^t-
\sum_x P(x)e^{-u(s|\calX|+t)P(x)}\right] \, \frac{\mathrm{d}u}{u},
\end{align}
which gives from \eqref{KT81c} and \eqref{KT81d} that the redundancy is given by
\begin{align}
&R_n \nonumber \\
&=\frac{1}{n}\sum_{t=0}^{n-1}\bE\left\{\ln\frac{(t+
s |\calX|)P(X_{t+1})}{N_t(X_{t+1})+s}\right\}\nonumber\\[0.1cm]
&=\frac{1}{n}\int_0^\infty\bigg(e^{-us}\sum_x P(x)
\sum_{t=0}^{n-1}[1-P(x)(1-e^{-u})]^t-\sum_x P(x)e^{-us |\calX| P(x)}
\sum_{t=0}^{n-1}e^{-uP(x)t}\bigg)
\, \frac{\mathrm{d}u}{u}\nonumber\\[0.1cm]
&=\frac{1}{n}\int_0^\infty \Biggl[e^{-us}
\sum_x \frac{1-[1-P(x)(1-e^{-u})]^n}{1-e^{-u}}-\sum_x \frac{P(x)
\, e^{-us|\calX|P(x)} \bigl(1-e^{-uP(x)n} \bigr)}{1-e^{-uP(x)}}\Biggr]
\, \frac{\mathrm{d}u}{u}\nonumber\\[0.1cm]
&= \frac{1}{n}\int_0^\infty\Biggl[
\frac{e^{-us}\bigl(|\calX|-\underset{x}{\sum}[1-P(x)(1-e^{-u})]^n\bigr)}{1-e^{-u}}-
\sum_x \frac{P(x) \, e^{-us|\calX|P(x)}(1-e^{-uP(x)n})}{1-e^{-uP(x)}}\Biggr]
\, \frac{\mathrm{d}u}{u}.\nonumber
\end{align}

\begin{figure}[h!t!b!]
\centering
\vspace*{-4.5cm}
\includegraphics[width=11cm]{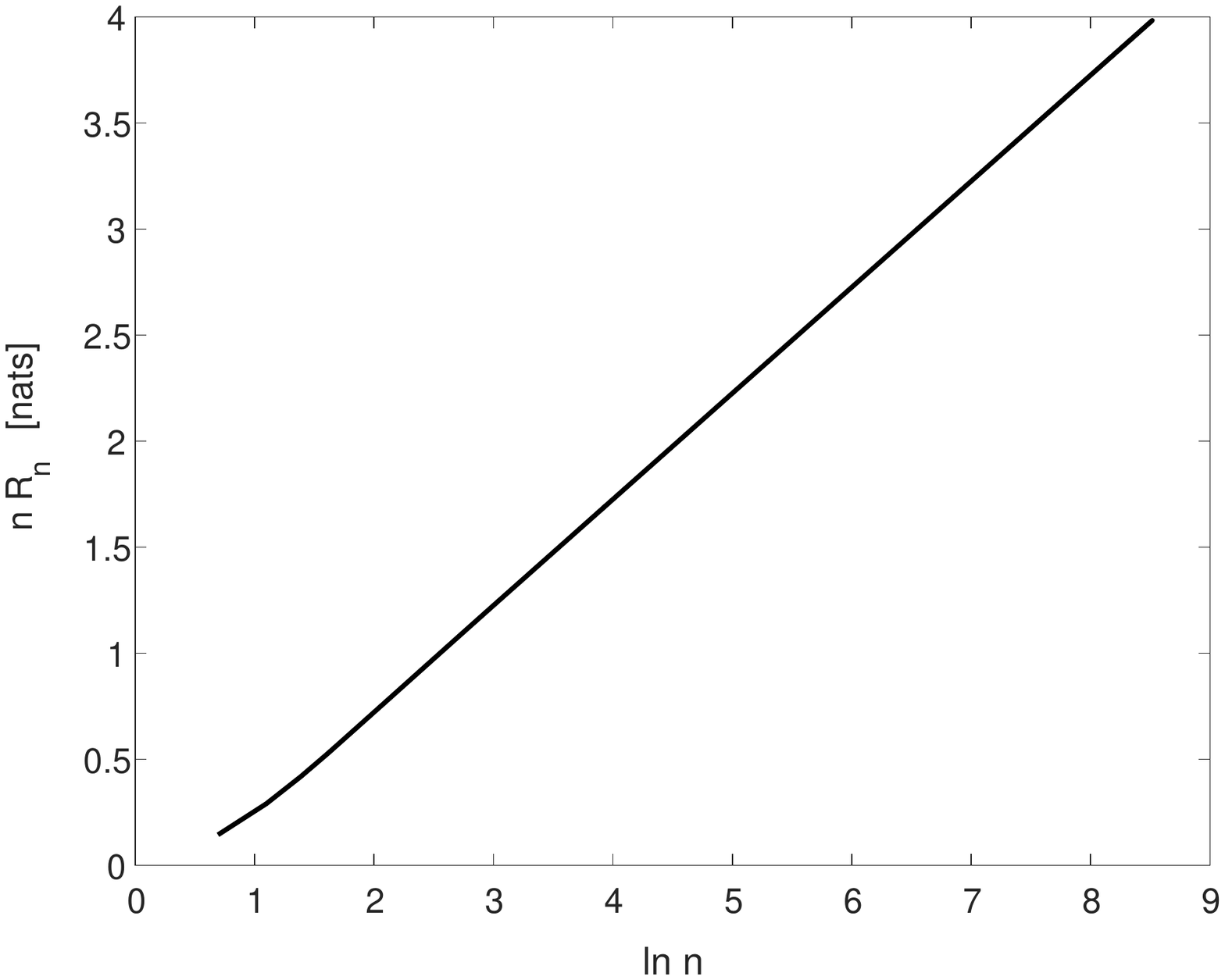}
\vspace*{-3.8cm}
\caption{The function $nR_n$ vs.\ $\ln n$ for the BSS and $s=\tfrac12$, in the
range $2\le n\le 5000$.}
\label{graph1}
\end{figure}

Fig.~\ref{graph1} displays $nR_n$ as a function of $\ln n$ for $s=\tfrac12$
in the range $1\le n\le 5000$. As can be seen, the graph is nearly a straight
line with slope $\tfrac12$, which is in agreement with the theoretical result
that $R_n\sim \frac{\ln n}{2n}$ (in nats per symbol) for large $n$ (see
\cite[Theorem~2]{KT81}).

\section{Summary and Outlook}
\label{outlook}

In this work, we have explored a well--known integral representation of
the logarithmic function, and demonstrated its applications in obtaining
exact formulas for quantities that involve expectations and second order
moments of the logarithm of a positive random variable (or the logarithm
of a sum of i.i.d. such random variables). We anticipate that this integral
representation and its variants can serve as a useful tool also in many
additional applications, as a rigorous alternative to the replica
method in some situations.

Our work in this paper focused on exact results. In future research, it
would be interesting to explore whether the integral representation we
have used is useful also in obtaining upper and lower bounds on expectations
(and higher order moments) of expressions that involves logarithms of
positive random variables. In particular, could the integrand of \eqref{ir}
be bounded from below and/or above in a non--trivial manner, that would
lead to new interesting bounds? Moreover, it would be even more useful if
the corresponding bounds on the integrand would lend themselves to closed--form
expressions of the resulting definite integrals.

Another route for further research relies on \cite[p.~363, Identity~(3.434.1)]{GR14},
which states that
\begin{align}
\label{eq:GR2}
\int_0^{\infty} \frac{e^{-\nu u} - e^{-\mu u}}{u^{\rho+1}} \, \mathrm{d}u =
\frac{\mu^\rho - \nu^\rho}{\rho} \cdot \Gamma(1-\rho), \qquad \mathrm{Re}(\mu)>0,
\; \mathrm{Re}(\nu)>0, \; \mathrm{Re}(\rho)<1.
\end{align}
Let $\nu:=1$, and $\mu := \overset{n}{\underset{i=1}{\sum}} X_i$ where
$\{X_i\}_{i=1}^n$ are positive i.i.d. random variables. Taking expectations
of both sides of \eqref{eq:GR2} and rearranging terms, gives
\begin{align} \label{rho-th moment}
\bE \biggl\{ \biggr( \sum_{i=1}^n X_i \biggr)^{\rho} \biggr\} = 1 +
\frac{\rho}{\Gamma(1-\rho)} \int_0^{\infty} \frac{e^{-u} - M_X^n(-u)}{u^{\rho+1}}
\; \mathrm{d}u, \qquad \forall \, \rho \in (0,1),
\end{align}
where $X$ is a random variable having the same density as of the $X_i$'s, and
$M_X(u) := \bE\bigl\{e^{uX}\bigr\}$ (for $u \in \reals$) denotes the MGF of $X$.
Since $\ln x = \underset{\rho \to 0}{\lim} \frac{x^\rho-1}{\rho}$ for $x>0$, it
follows that \eqref{rho-th moment} generalizes \eqref{Elnsum} for the logarithmic
expectation. Identity~\eqref{rho-th moment}, for the $\rho$-th moment of a sum
of i.i.d. positive random variables with $\rho \in (0,1)$, may be used in some
information-theoretic contexts rather than invoking Jensen's inequality.

\subsection*{Acknowledgment}
The authors are thankful to Zbigniew Golebiewski and Cihan Tepedelenlio\v{g}lu
for bringing references \cite{Knessl98} and \cite{RajanT15}, respectively, to
their attention.

\end{document}